\documentclass[aps,prl,reprint,twocolumn,superscriptaddress]{revtex4-1}

\usepackage{amsmath,amsfonts,amsthm,graphicx,mathtools,epstopdf,hyperref,tabularx}
\usepackage[caption=false]{subfig}

\newcommand{\ceil}[1]{\ensuremath{\left\lceil #1 \right\rceil}} 
\newcommand{\floor}[1]{\ensuremath{\left\lfloor #1 \right\rfloor}} 
\newcommand{\ket}[1]{\ensuremath{\left| #1 \right>}}

\newcommand{\expval}[1]{\ensuremath{\left< #1 \right>}} 
\newcommand{\multinom}[2]{\ensuremath{\left( \begin{array}{c} #1 \\ #2 \end{array} \right)}} 

\newcommand{\pset}[1]{\ensuremath{\mathcal{P}_{#1}}}
\newcommand{\lrset}[1]{\ensuremath{\mathcal{L}_{#1}}}
\newcommand{\qset}[1]{\ensuremath{\mathcal{Q}_{#1}}}
\newcommand{\ns}[1]{\ensuremath{\mathcal{N}_{#1}}}
\newcommand{\mdlset}[1]{\ensuremath{\mathcal{M}_{#1}}}
\newcommand{\proj}[0]{\ensuremath{P}} 

\newcommand{\cg}[1]{\ensuremath{c_{#1}}} 
\newcommand{\putz}[2]{\ensuremath{M(#1,#2)}}

\newcommand{\sci}[2]{\ensuremath{#1 \times 10^{#2}}} 

\newtheorem{proposition}{Proposition}
\newtheorem*{corollary*}{Corollary}

\begin{document}

\title{\textbf{Measurement-dependent locality beyond i.i.d.}}
\author{Ernest Y.-Z.\ Tan}
\affiliation{Department of Physics, National University of Singapore, 2 Science Drive 3, 117542 Singapore, Singapore}
\author{Yu Cai}
\affiliation{Centre for Quantum Technologies, National University of Singapore, 3 Science Drive 2, Singapore 117543, Singapore}
\author{Valerio Scarani}
\affiliation{Department of Physics, National University of Singapore, 2 Science Drive 3, 117542 Singapore, Singapore}
\affiliation{Centre for Quantum Technologies, National University of Singapore, 3 Science Drive 2, Singapore 117543, Singapore}

\begin{abstract}
When conducting a Bell test, it is normal to assume that the preparation of the quantum state is independent of the measurements performed on it. Remarkably, the violation of local realism by entangled quantum systems can be certified even if this assumption is partially relaxed. Here, we allow such measurement dependence to correlate multiple runs of the experiment, going beyond previous studies that considered independent and identically distributed (i.i.d.) runs. To do so, we study the polytope that defines \textit{block-i.i.d.}\ measurement-dependent local models. We prove that non-i.i.d.\ models are strictly more powerful than i.i.d.\ ones, and comment on the relevance of this work for the study of randomness amplification in simple Bell scenarios with suitably optimised inequalities. 
\end{abstract}

\maketitle

\section{Introduction}
Since their introduction, by John Bell in 1964~\cite{bell64}, Bell inequalities have been a subject of extensive study, as they highlight the fact that quantum theory is incompatible with local realism. Numerous experimental tests of Bell inequalities have been carried out, with the results being overwhelmingly in favour of the quantum predictions. Of particular note are the recent loophole-free Bell tests~\cite{hensen15,shalm15,giustina15}, which simultaneously addressed several loopholes that had been raised regarding previous experiments. All these tests were conducted under the assumption that the choice of measurements and the state of the source are independent in each run. This observation should not be taken as a reservation: such \textit{measurement independence} is an essential piece of the scientific method, and its negation would be rightly considered conspiratorial. This makes it all the more remarkable that quantum theory can be proved incompatible with local realism even if this assumption is relaxed to some extent.

Indeed, while unrestricted \textit{measurement dependence} would lead to an unfalsifiable superdeterminism~\cite{brans88}, it was noted by Hall~\cite{hall11} that the violation of Bell inequalities keeps its meaning if some restrictions are made. This led to the study of measurement-dependent local (MDL) scenarios, where some correlation is allowed between the measurement choices and the source. A few subsequent works refined our understanding of measurement dependence \cite{barrett11,koh12,pope13,yuan15}, all sticking to known inequalities. A significant breakthrough was achieved when P{\"u}tz and coworkers noticed that the traditional Bell inequalities are no longer optimal: other linear constraints, suitably named \textit{MDL inequalities}, more tightly define the conditions under which local realism holds in the MDL scenario. Their works \cite{putz14,putz15} developed the mathematical framework to study these inequalities. Their most celebrated discovery is the following: there exist quantum correlations that violate local realism with ``arbitrarily low measurement independence'', that is, as long as the MDL model does not trivially allow us to reproduce all no-signalling correlations. The corresponding inequality has been tested in an experiment \cite{putz16}. 

Measurement dependence in Bell-type tests is also central in the task of \textit{randomness amplification}, where one aims to turn a single weak source of randomness (one in which the subsequent outcomes may be correlated in an almost unrestricted way) into a perfect coin. This task is provably impossible with classical information processing, but it becomes possible if the weak source is used to choose the inputs (including the state) in a Bell test, whose outcomes are taken as the new random numbers~\cite{colbeck12,gallego13,brandao16,bouda14,chung15}. One may wonder why the optimised approach of P{\"u}tz and coworkers has not yet been applied to improve the bounds on randomness amplification. The reason is that, as reported so far, that approach has been developed under the assumption that the runs are independent and identically distributed (i.i.d.); and amplification of an i.i.d.\ source is trivial \footnote{If a source is guaranteed to be i.i.d., one can create two independent sources from its output; and from two independent sources a perfect coin can always be extracted, at least in principle.}.

In this paper we study MDL for \textit{block-i.i.d.\ models}, in which, as the name indicates, blocks of $N$ runs are i.i.d.\ but the $N$ runs in each block can be arbitrarily correlated \cite{pope13,yuan15}. Among the results (see Table \ref{table_summary} for a comprehensive overview), we prove that MDL models in fact become strictly more powerful if the i.i.d.\ assumption is dropped. This was not a foregone conclusion: under measurement independence, the local bound of a Bell inequality is the same with or without the i.i.d.\ assumption, only the estimates of finite-sample fluctuations differ.

\section{Measurement dependence}

\begin{figure*}
\centering
\subfloat[Measurement-dependent locality (i.i.d.)]{
	\includegraphics[height=2.5cm, keepaspectratio, trim=0.6cm 0.2cm 10.6cm 0.4cm, clip=true]{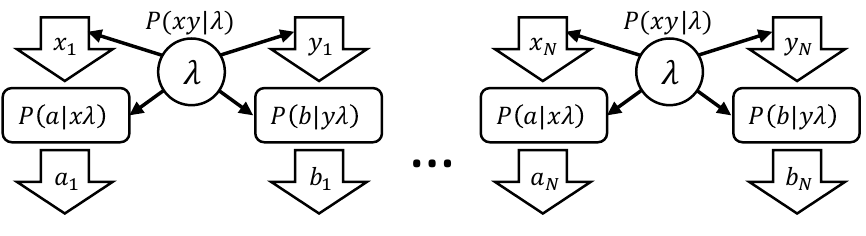}
	\label{fig_schematic1}
} \hspace{.4cm}
\subfloat[Measurement-dependent locality (block-i.i.d.)]{
	\includegraphics[height=2.5cm, keepaspectratio, trim=11.3cm 0.2cm 1.3cm 0.4cm, clip=true]{schematic}
	\label{fig_schematic2}
}
\caption{Measurement-dependent local models of a Bell test with $N$ runs, where the variable $\lambda$ determines the local output probabilities $P(a|x \lambda)$ and $P(b|y \lambda)$, and is also correlated with the inputs via $P(xy| \lambda)$. The i.i.d.\ case is shown in (a), where all the runs are independent and the same probabilities $P(a|x \lambda)$, $P(b|y \lambda)$, $P(xy| \lambda)$ are used in each run. The block-i.i.d.\ case is shown in (b), where an entire stretch of $N$ runs is modelled in parallel as a single block.}
\end{figure*}

Let us begin by reviewing how measurement dependence is formalised in the i.i.d.\ case. Consider a bipartite Bell test setup with two experimenters, Alice and Bob. Alice's measurement choice is specified by an input $x \in \left\{0, 1, ..., d_X-1 \right\}$, and she obtains an output $a \in \left\{0, 1, ..., d_A-1 \right\}$. Similarly, Bob's input and output will be labelled $y \in \left\{0, 1, ..., d_Y-1 \right\}$ and $b \in \left\{0, 1, ..., d_B-1 \right\}$ respectively. In the case of i.i.d.\ runs, Alice and Bob can measure the probabilities $P(abxy)$ of their various inputs and outputs occurring in a single experimental run. The set of all valid probability distributions $P(abxy)$ will be denoted $\pset{1}$, the single-run probability space \footnote{When measurement independence is taken for granted, the probabilities $P(xy)$ play a trivial role, so one normally uses $P(ab|xy)$ instead of $P(abxy)$ to discuss Bell inequalities. The two descriptions are equivalent for any fixed set of nonzero values for $P(xy)$, since in that case the conversion between the two is an invertible linear transformation; whereas if $P(xy)=0$ for some $(x,y)$, it will be impossible to gather the data to reconstruct $P(ab|xy)$ anyway. Throughout this paper we consider $P(xy)$ (or $P(\vec{x}\vec{y})$ for the block-i.i.d\ case) to be known, corresponding to taking a slice of $\pset{1}$. A particularly useful choice is the uniform-measurements slice, $P(xy) = 1/(d_X d_Y)$ for all $(x,y)$. This is discussed further in the Supplemental Material.}. Later, we will consider probability spaces $\pset{N}$, which account for input-output combinations over $N$ runs. 

Local realistic models are defined as those that admit a decomposition
\begin{align}
P(abxy) = \left(\int d\lambda \, w(\lambda) P(a|x\lambda) P(b|y\lambda)\right) P(xy),
\label{lrdef}
\end{align}
where $\lambda$ is a ``local hidden variable'' or ``strategy'' that determines the conditional probabilities of the outputs given the inputs. For any fixed $P(xy)$, the set of points in $\pset{1}$ admitting a local realistic model as described above forms a polytope~\cite{fine82,brunner14}, denoted as $\lrset{1}$. The linear inequalities satisfied by all points in $\lrset{1}$ are the Bell inequalities. In MDL models, the input probabilities can be conditioned on $\lambda$ as well (Fig~\ref{fig_schematic1}), so the achievable probability distributions are of the form
\begin{align}
P(abxy) = \int d\lambda \, w(\lambda) P(a|x\lambda) P(b|y\lambda) P(xy|\lambda).
\label{mdldef}
\end{align}
Clearly, this is a superset of $\lrset{1}$. If no constraints are imposed on the conditional input probabilities $P(xy|\lambda)$, MDL models can trivially reproduce all quantum distributions~\cite{brans88,putz14}. To exclude this scenario, one approach~\cite{putz14,putz15} is to impose linear bounds 
\begin{align}
l \leq P(xy|\lambda) \leq h.
\label{mdlconstraints}
\end{align}
Let us notice that in the language of randomness, a source characterised only by $h<1$ would be called an \textit{i.i.d.\ min-entropy source}. Given such constraints, the set of points in $\pset{1}$ achievable by MDL models is a polytope as well~\cite{putz14,putz15}, which we shall denote as $\mdlset{1}$. The linear inequalities satisfied by all points in the MDL polytope are the MDL inequalities. While famously the local polytope and the quantum set are subsets of the  no-signalling polytope $\ns{1}$, which is a slice of $\pset{1}$ defined by suitable linear equality constraints~\cite{brunner14}, the MDL polytope extends into the signalling region~\cite{thinh13,putz14}. This contributes to making its characterisation, if not conceptually harder, certainly computationally heavier.

It is important to become familiar with the constraints \eqref{mdlconstraints}, so we make a few remarks about them. First, the normalisation $\sum_{xy}P(xy|\lambda)=1$ implies $l \leq 1/(d_X d_Y) \leq h$, because there are $d_X d_Y$ combinations of inputs in the Bell test. If either $l$ or $h$ are set at $1/(d_X d_Y)$, then this enforces $P(xy|\lambda)=1/(d_X d_Y)$ which is the case of measurement independence. A deterministic choice of inputs conditioned on $\lambda$ would be allowed by $h=1$ and $l=0$, but one does not need to go all the way to determinism for the MDL scenario to become trivial. In particular, $l=0$ already means that there can exist one or more pairs of inputs $(x,y)$ that are never used, and this is very powerful. For instance, consider the 2-input 2-output case $(d_X,d_Y,d_A,d_B)=(2,2,2,2)$: as soon as one pair of settings is not used, one can use local variables to fake any no-signalling distribution. Therefore, the values $h=1/3$ and $l=0$ already describe a trivial situation in which the violation of local realism cannot possibly be certified. As mentioned in the introduction, a key finding of P\"utz and coworkers ~\cite{putz14} is an MDL inequality that admits a quantum violation for any $l>0$ in the i.i.d.\ case; in particular, this statement also holds for all $h < 1/3$ even if $l$ is left unspecified.

\begin{table*}[h!t!]
\caption{Summary of results for MDL$_N$ models. All the results have been obtained for $(d_X,d_Y,d_A,d_B)=(2,2,2,2)$ and on the slice $P(\vec{x}\vec{y})=1/4^N$. For comparison, the main result of P\"utz and coworkers \cite{putz14,putz15} reads $\qset{1}\subseteq \mdlset{1}\,\Leftrightarrow\,h \geq 1/3$.}
\def\arraystretch{2}
\setlength\tabcolsep{.2cm}
\begin{center}
\begin{tabularx}{.95\textwidth}{|c|X|X|}
\hline
& \multicolumn{1}{c|}{\textbf{Full probabilities $\mathbf{P}(\vec{\mathbf{a}}\vec{\mathbf{b}}\vec{\mathbf{x}}\vec{\mathbf{y}})$ in $\mathbf{\pset{N}}$}} & \multicolumn{1}{c|}{\textbf{Coarse-grained probabilities $\mathbf{P(abxy)}$ in $\mathbf{\pset{1}}$}} \\ 
\hline
$\mdlset{N}$ & $\ns{1}^{\times 2} \subseteq \mdlset{2}\,\Leftrightarrow\,h \geq 1/\sqrt{10}$. Also, $\qset{1}^{\times 2} \not\subseteq \mdlset{2}$ at least for $h < 1/\sqrt{13}$. & $\cg{2}(\ns{1}^{\times 2}) \subseteq \cg{2}(\mdlset{2})\,\Leftrightarrow\,h \geq 1/\sqrt{10}$. This implies $\cg{N}(\qset{N}) \subseteq \cg{2}(\mdlset{2})\,\Leftarrow\,h \geq 1/\sqrt{10}$ for all $N$. Also, $\qset{1} \not\subseteq \cg{2}(\mdlset{2})$ at least for $h < 1/\sqrt{14}$.\\
\hline
$\mdlset{N}'$ & $\ns{1}^{\times 2}, \qset{1}^{\times 2} \subseteq \mdlset{2}'\,\Leftrightarrow\,h \geq 1/3$. & $\qset{1} \not\subseteq \cg{2}(\mdlset{2}')$ at least for $h < 1/\sqrt{12}$. Ref.~\cite{pope13} proved $\qset{1} \not\subseteq \cg{N}(\mdlset{N}')$ at least for $h \lesssim 0.258$.\\
\hline
\end{tabularx}
\end{center}
\def\arraystretch{1}
\label{table_summary}
\end{table*}

Having introduced these notions, following Pope and Kay \cite{pope13} we generalise them to \textit{block-$N$-i.i.d.\ MDL models} (or MDL$_N$ models for short), which are the focus of this work. We now consider blocks of $N$ experimental runs in parallel, dealing with $N$-tuples of inputs and outputs $\vec{x},\vec{y},\vec{a},\vec{b}$ (Fig.~\ref{fig_schematic2}). After many repetitions of these $N$ runs, one can reconstruct $P(\vec{a}\vec{b}\vec{x}\vec{y})$; the set of all valid probability distributions of this form is denoted by $\pset{N}$. 

Within this set, MDL$_N$ models achieve probability distributions of the form
\begin{align}
P(\vec{a} \vec{b} \vec{x} \vec{y}) = \int d\lambda \, w(\lambda) P(\vec{a} | \vec{x} \lambda) P(\vec{b} | \vec{y} \lambda) P(\vec{x} \vec{y} | \lambda).
\label{mdlndef}
\end{align}

\noindent Similar to the i.i.d.\ case, one can impose linear constraints $L_{} \leq P(\vec{x} \vec{y} | \lambda) \leq H_{}$ on the MDL$_N$ model. Under such constraints, the set $\mdlset{N}$ of all points in $\pset{N}$ attainable by MDL$_N$ models is again a polytope. This can be shown simply by noticing that this MDL$_N$ scenario for $(d_X,d_Y,d_A,d_B)$ is mathematically equivalent to the MDL$_1$ scenario for $(d_X^N,d_Y^N,d_A^N,d_B^N)$.

In this work, we focus on the case where the lower bound $L_{}$ is left unspecified, with larger values of $H_{}$ corresponding to greater amounts of measurement dependence. To facilitate comparison with the i.i.d.\ case, we denote $H_{}^{1/N}\equiv h$, so finally we are going to work with the constraint
\begin{align}
P(\vec{x} \vec{y} | \lambda) \leq h^N.
\label{minentropy}
\end{align}
In the language of randomness, Eq.~\eqref{minentropy} says that the inputs $(\vec{x}, \vec{y})$ are drawn from a \textit{block-i.i.d.\ min-entropy source} with $N \log_2 (1/h)$ bits of input entropy per use, which is strictly more general than $N$ uses of an i.i.d.\ min-entropy source with $\log_2 (1/h)$ bits of input entropy per use.

Before presenting our results, we need to introduce two more notions. The first is \textit{single-run coarse-graining}, which converts the block probabilities $P(\vec{a} \vec{b} \vec{x} \vec{y})$ into average single-run probabilities $P(abxy)$. The function $\cg{N}: \pset{N} \to \pset{1}$ that represents this coarse-graining is linear, and hence maps polytopes to polytopes; it is described in detail in the Supplemental Material. Obviously, information is lost in this procedure, but the resulting probability space is of considerably lower dimension and hence easier to study. Besides, if a violation of local realism is seen in the coarse-grained version, it must also be present in the full probabilities. Finally, for large $N$, it will be hard to reconstruct the $P(\vec{a} \vec{b} \vec{x} \vec{y})$ from the experimental data.

The second notion is that of restricting MDL$_N$ models to local strategies that are \textit{independent but not identically distributed}. This is obtained by assuming
\begin{align}
P(\vec{a} | \vec{x} \lambda) = \prod_{j=1}^N P(a_j | x_j \lambda), \, P(\vec{b} | \vec{y} \lambda) = \prod_{j=1}^N P(b_j | y_j \lambda).
\label{indrundef}
\end{align} in \eqref{mdlndef}. The corresponding set of probabilities still forms a polytope (see Supplemental Material), which we denote as $\mdlset{N}'$. This is useful for comparison with the works of Pope and Kay \cite{pope13} and P\"utz and coworkers \cite{putz14,putz15}. 

\section{Results}

Our results are obtained for the case $(d_X,d_Y,d_A,d_B)=(2,2,2,2)$ and on the slice $P(\vec{x}\vec{y})=1/4^N$, describing the natural assumption that all the inputs appear uniformly distributed when there is no information on $\lambda$. The results are listed in Table \ref{table_summary}, with detailed proofs given in the Supplemental Material. Here we comment on them.

Firstly, most of these results relate an MDL$_N$ scenario with \textit{product} sets of the type $\ns{1}^{\times N}$ or $\qset{1}^{\times N}$ (see Supplemental Material), defined by a condition similar to that in Eq.~\eqref{indrundef}. The reason for this choice goes back to a previous remark: $\mdlset{N}(2,2,2,2)=\mdlset{1}(2^N,2^N,2^N,2^N)$. By studying the most general quantum statistics $\qset{N}$, we would actually be discussing MDL$_1$ for larger alphabets. In order to give our study a clear flavour of going beyond i.i.d., therefore, \textit{we discuss the power of MDL$_N$ models to reproduce statistics achievable with independent entangled pairs}. In other words, we are addressing the following question: by implementing a ``routine'' Bell test with independent entangled pairs, up to which value of $h$ can one obtain a probability distribution that falsifies local realism, even under an MDL$_N$ assumption? We note that the pairs do not need to be identically distributed, which is a pleasant feature for comparison with experiments, in which some parameters may drift with time.

We start with the \textit{left column of Table \ref{table_summary}}, dealing with the general $P(\vec{a} \vec{b} \vec{x} \vec{y})$. Unfortunately, $\mdlset{2}$ already has up to \sci{5.24}{9} vertices, making it impractical to compute all its facets. By instead exploiting properties of the Popescu-Rohrlich (PR) box~\cite{popescu94}, we have been able to prove that $\ns{1}^{\times 2}$ is already enclosed by $\mdlset{2}$ for $h =1/\sqrt{10}$ (top-left corner). In particular, then, it will be impossible for $\qset{1}^{\times 2}$ to violate local realism all the way up to $h=1/3$ in the MDL$_2$ scenario. We found that the MDL$_1$ inequality that was violated in the whole non-trivial range of $h$~\cite{putz14,putz16} loses much of its robustness under MDL$_2$: it can no longer be violated by the quantum distribution specified in Refs.~\cite{putz14,putz16} when $h \gtrsim 0.255$ (see Supplemental Material for details). We were still able to show that $\qset{1}^{\times 2} \not\subseteq \mdlset{2}$ for all $h < 1/\sqrt{13}$, but it remains an open question whether other points in $\qset{1}^{\times 2}$ can violate local realism for higher values, possibly up to $h = 1/\sqrt{10}$. 

The bottom-left corner shows that the robustness reported by P\"utz and coworkers is recovered if the two runs are constrained to be independent as in Eq.~\eqref{indrundef}. This result, though maybe not surprising, does constitute a generalisation of the original one, insofar as the runs are not required to be identical.

Moving to the \textit{right column of Table \ref{table_summary}}, we deal with the coarse-grained probabilities $P(abxy)$. In the upper-right corner, the new piece of information is that $\cg{2}(\ns{1}^{\times 2}) \subseteq c_2(\mdlset{2})\,\Rightarrow\,h \geq 1/\sqrt{10}$, while the converse implication follows from the previous result. More interestingly, here we are able to make a statement for any $N$, and about $\qset{N}$ rather than only $\qset{1}^{\times N}$, because it can be shown that $\cg{N}(\qset{N}) \subseteq \ns{1} = \cg{N}(\ns{1}^{\times N})$ (see Supplemental Material). Therefore, after coarse-graining, no quantum statistics will show any violation of local realism for $h \geq 1/\sqrt{10}$ under the MDL$_{N>1}$ assumption. We do not know if this bound is tight, but we find at least that there exists a point in $\qset{1}$ which remains outside $c_2(\mdlset{2})$ for all $h < 1/\sqrt{14}$.

Finally, the lower-right corner is the most constrained situation, that was studied by Pope and Kay~\cite{pope13}~\footnote{This situation was also studied in Ref.~\cite{yuan15}, for the case where Alice and Bob's inputs are uncorrelated when conditioned on $\lambda$, $P(\vec{x}\vec{y}|\lambda) = P(\vec{x}|\lambda)P(\vec{y}|\lambda)$. Even under this restriction, the threshold value of $h$ is only slightly higher.}. They considered the CHSH inequality~\cite{clauser69} and proved that it is violated by $\cg{N}(\qset{1}^{\times N})$ for any $N$ up to $h \lesssim 0.258$ (denoted $P_\infty$ in their paper); when $N=2$ in particular, it is violated up to $h \lesssim 0.280$. For this scenario with $N=2$, we have found points in $\qset{1}$ that violate local realism for $h<1/\sqrt{12}$, but cannot make conclusive statements for larger values of $h$. 

\section{Conclusion}

We have studied block-i.i.d.\ models for measurement-dependent locality (MDL$_N$), and their power to reproduce statistics that can be produced with $N$ independent entangled pairs. The MDL$_N$ model is the least constrained one: a weak random source with min-entropy $N\log(1/h)$. For specific results, we have considered the Bell scenario with two inputs and two outputs per party. For $N=1$, it was known that MDL models become too powerful at $h=1/3$, and remarkably, quantum correlations could demonstrate violation of local realism all the way up to that value. However, this conclusion does not stand when the i.i.d.\ assumption is relaxed: already for $N=2$ we have shown that the threshold value is reduced to $h =1/\sqrt{10}$; and with correlations achievable with two entangled pairs we have not been able to find any violation beyond $h =1/\sqrt{13}$. We have obtained similar results for more restricted MDL models and for a coarse-grained data processing, some of which are valid for arbitrary $N$.

We finish by commenting on the implications of our results for randomness amplification. The first results~\cite{colbeck12,gallego13,brandao16} were obtained for a slightly stronger model of random sources, the so-called Santha-Vazirani sources, but subsequent results have claimed the possibility of amplifying even a min-entropy source \cite{bouda14,chung15}. However, all these protocols require multi-partite entanglement and are not robust to deviations from an ideal quantum state; it is currently an open problem to devise a randomness amplification protocol that can be implemented with existing devices. The MDL approach started by P\"utz and coworkers gives the hope of deriving such a protocol: robust, and for the simplest Bell scenario. Our paper is the first step in this direction, but we are not yet there. A computational study of MDL$_N$ for larger $N$ would be challenging, so one would have to try obtaining analytical results instead.

\section*{Acknowledgments}

We acknowledge useful correspondence with Gilles P{\"u}tz and Nicolas Gisin, particularly regarding the computational methods used to study the MDL polytope. This work is funded by the Singapore Ministry of Education (partly through the Academic Research Fund Tier 3 MOE2012-T3-1-009), by the National Research Foundation of Singapore, Prime Minister's Office, under the Research Centres of Excellence programme.

\bibliographystyle{apsrev4-1} 
\bibliography{mdl_blockiid_biblio} 

\clearpage

\appendix

\onecolumngrid

\section{\large Supplemental Material}

In this Supplemental Material, we begin by laying out a framework for block-i.i.d.\ models, defining the quantum sets and no-signalling sets of interest in the block-i.i.d.\ case. We then show that the MDL$_N$ set is a polytope, and give the form of its vertices. This allows us to study the i.i.d.\ MDL inequality described in Refs.~\cite{putz14,putz15} in the MDL$_N$ scenario, and show that it becomes substantially less robust. Finally, we describe several techniques that can be used to study the MDL$_N$ polytope, most importantly the notion of \textit{$k$-mismatch strategies}. Using these techniques, we derive the results shown in Table~\ref{table_summary} in the main text. 

\section{The probability space and input probabilities}

In this section, we discuss the block-$N$-i.i.d.\ scenario directly, with the i.i.d.\ scenario being described by the $N=1$ case. We begin by noting that since $\pset{N}$ is defined by a set of linear equality and inequality constraints on a vector space, it forms a polytope. In this work, we choose to use the definition of a polytope as an intersection of finitely many half-spaces, and all polytopes will be implicitly assumed to be compact. This is then equivalent to defining a polytope as a convex hull of finitely many points. 

As mentioned in the main text, the probabilities $P(\vec{a}\vec{b}\vec{x}\vec{y})$ can be converted into conditional probabilities $P(\vec{a}\vec{b}|\vec{x}\vec{y})$ by dividing by $P(\vec{x}\vec{y})$, assuming all values of $P(\vec{x}\vec{y})$ are nonzero. This is not a linear transformation on $\pset{N}$ as a whole; however, on any slice of $\pset{N}$ specified by fixing the values of $P(\vec{x}\vec{y})$, it is indeed an invertible linear transformation. This allows us to directly apply many theorems derived in terms of conditional probabilities $P(\vec{a}\vec{b}|\vec{x}\vec{y})$ to the full probabilities $P(\vec{a}\vec{b}\vec{x}\vec{y})$. An alternative approach for future work may be to work entirely with the conditional probabilities, in which case the local realistic set and quantum set have been extensively characterised, but the structure of the MDL$_N$ set becomes less clear. 

There are multiple ways to compute these input probabilities $P(\vec{x}\vec{y})$; for instance, given the set of probabilities $P(\vec{a}\vec{b}\vec{x}\vec{y})$, they can be computed by summing over $\vec{a}$ and $\vec{b}$. This is used to convert $P(\vec{a}\vec{b}\vec{x}\vec{y})$ to conditional probabilities $P(\vec{a}\vec{b}|\vec{x}\vec{y})$. Alternatively, they could be computed as $P(\vec{x}\vec{y}) = \int d\lambda \, w(\lambda) P(\vec{x}\vec{y}|\lambda)$ if a $\lambda$-strategy is specified. As a consistency check, we note that these methods are equivalent, since
\begin{align}
\sum_{\vec{a},\vec{b}} P(\vec{a}\vec{b}\vec{x}\vec{y}) &= \sum_{\vec{a},\vec{b}} \int d\lambda \, w(\lambda) P(\vec{a}\vec{b}|\vec{x}\vec{y}\lambda) P(\vec{x}\vec{y}| \lambda) \nonumber \\ 
&= \int d\lambda \, w(\lambda) \left(\sum_{\vec{a},\vec{b}} P(\vec{a}\vec{b}|\vec{x}\vec{y}\lambda)\right) P(\vec{x}\vec{y}| \lambda) \nonumber \\ 
&= \int d\lambda \, w(\lambda) P(\vec{x}\vec{y}| \lambda) \quad \text{by normalisation}.
\end{align}

\noindent We also note that if the average single-run probabilities $P(xy)$ are needed, they could be computed by averaging the probabilities $P(\vec{x}\vec{y})$ directly, or by first computing the probabilities $P(abxy)$ by averaging $P(\vec{a}\vec{b}\vec{x}\vec{y})$, then summing over $a$ and $b$. All these methods can be shown to give the same value. This raises  some question of whether to use $P(\vec{x}\vec{y})$ or $P(xy)$ to specify a slice of $\pset{N}$. However in this work, we will restrict ourselves to the former, as the latter specification results in a slice where the conversion between $P(\vec{a}\vec{b}\vec{x}\vec{y})$ and $P(\vec{a}\vec{b}|\vec{x}\vec{y})$ is not necessarily linear.

\section{The coarse-graining function}

Starting with the block-2-i.i.d.\ case, we define the coarse-graining function as follows: given any point $p \in \pset{2}$ corresponding to probabilities $P_p(a_1 a_2 b_1 b_2 x_1 x_2 y_1 y_2)$, we define $\cg{2}(p) \in \pset{1}$  as the point corresponding to probabilities
\begin{align}
P_{\cg{2}(p)}(abxy) &= \frac{1}{2} \left(P_{A_1 B_1 X_1 Y_1}(abxy) + P_{A_2 B_2 X_2 Y_2}(abxy) \right)\\
&= \frac{1}{2} \left( \sum_{a_2, b_2, x_2, y_2} P_p(a a_2 b b_2 x x_2 y y_2) + \sum_{a_1, b_1, x_1, y_1} P_p(a_1 a b_1 b x_1 x y_1 y) \right).
\label{cgrain2def}
\end{align}

\noindent This represents an average of the probabilities in the first and second runs of getting the input-output combination $(a,b,x,y)$. The generalisation of the coarse-graining function to the block-$N$-i.i.d.\ case is fairly straightforward, though cumbersome to express:
\begin{align}
P_{\cg{N}(p)}(abxy) = \frac{1}{N} \, \sum_{j=1}^N \left( \sum_{(\vec{a},\vec{b},\vec{x},\vec{y}) \in S_j(a,b,x,y) } P_p(\vec{a} \vec{b} \vec{x} \vec{y}) \right),
\label{cgraindef}
\end{align}

\noindent where $S_j(a,b,x,y)$ is defined as the set of tuples $(\vec{a},\vec{b},\vec{x},\vec{y})$ such that $(a_j,b_j,x_j,y_j) = (a,b,x,y)$, using $a_j$ to denote the $j^\text{th}$ entry of $\vec{a}$ and so on. It can be seen that the function is linear, and hence admits a matrix representation for computational purposes. It is not injective as a function from $\pset{N}$ to $\pset{1}$, because it is easy to find two points $p, p' \in \pset{N}$ such that $\cg{N}(p) = \cg{N}(p')$, for instance by permuting the order of the runs. This is consistent with its interpretation as a coarse-graining which loses some information about the original point. Another fairly intuitive property of the coarse-graining function is as follows:

\begin{proposition}
\label{prop_cg}
Consider any point $p \in \pset{N}$ that corresponds to the repetition of a single-run probability distribution $p_1 \in \pset{1}$ over $N$ runs,
\begin{align}
P_p(\vec{a} \vec{b} \vec{x} \vec{y}) = \prod_{j=1}^{N} P_{p_1}(a_j b_j x_j y_j).
\end{align}

\noindent Then $\cg{N}(p) = p_1$.
\end{proposition}

\begin{proof}

Referring to Eq.~\ref{cgraindef} for the definition of the coarse-graining function, we see that for the specified point $p$, the first term in the summation has the form
\begin{align}
\sum_{(\vec{a},\vec{b},\vec{x},\vec{y}) \in S_1(a,b,x,y) } P_p(\vec{a} \vec{b} \vec{x} \vec{y}) &= \sum_{(\vec{a},\vec{b},\vec{x},\vec{y}) \in S_1(a,b,x,y) } \left( \prod_{j=1}^{N} P_{p_1}(a_j b_j x_j y_j) \right) \nonumber \\
&= \sum_{a_2, b_2, x_2, y_2} ... \sum_{a_N, b_N, x_N, y_N} \left( P_{p_1}(abxy) \prod_{j=2}^{N} P_{p_1}(a_j b_j x_j y_j) \right) \nonumber \\
&= P_{p_1}(abxy) \prod_{j=2}^{N} \left( \sum_{a_j, b_j, x_j, y_j} P_{p_1}(a_j b_j x_j y_j) \right) \nonumber \\
&= P_{p_1}(abxy) \text{ \quad by normalisation of $P_{p_1}$},
\end{align}

\noindent and similarly for the other terms as well. Therefore, 
\begin{align}
P_{\cg{N}(p)}(abxy) &= \frac{1}{N} \sum_{j=1}^{N} P_{p_1}(abxy) = P_{p_1}(abxy),
\end{align}

\noindent which is the result to be proven.
\end{proof}

\section{The quantum set and no-signalling set}

As stated in the main text, when considering quantum models, we shall mainly consider product sets $\qset{1}^{\times N} \subseteq \pset{N}$, defined by the set of points in $\pset{N}$ admitting a decomposition
\begin{align}
P(\vec{a} \vec{b} \vec{x} \vec{y}) = \left( \int d\lambda \, w(\lambda) \prod_{j=1}^N \text{Tr} \left[ \rho^{(j)}(\lambda) \left(\proj_{a_j|x_j}^{(j)}(\lambda) \otimes \proj_{b_j|y_j}^{(j)}(\lambda) \right) \right] \right) P(\vec{x} \vec{y}),
\label{qsetn}
\end{align}

\noindent where the measurements $\proj_{a_j|x_j}^{(j)}(\lambda)$, $\proj_{b_j|y_j}^{(j)}(\lambda)$ can be considered projective without loss of generality because the dimension of the systems is left unconstrained. This describes the situation where the quantum states and measurements may differ from one run to the other but are independent across the runs, similar to independent-runs MDL$_N$ models (Eq.~\eqref{indrundef}). We contrast this to the most general quantum set $\qset{N}$, given by points of the form
\begin{align}
P(\vec{a} \vec{b} \vec{x} \vec{y}) = \left( \int d\lambda \, w(\lambda) \text{Tr} \left[ \rho(\lambda) \left(\proj_{\vec{a}|\vec{x}}(\lambda) \otimes \proj_{\vec{b}|\vec{y}}(\lambda) \right) \right] \right) P(\vec{x} \vec{y}).
\label{qsetncoherent}
\end{align}

\noindent This allows for coherent quantum states and measurements across multiple runs, and clearly $\qset{1}^{\times N} \subseteq \qset{N}$. However, we mostly do not consider $\qset{N}$ in this work, because it would be mathematically identical to an i.i.d.\ scenario with a larger alphabet. In addition, it would also be difficult to implement experimentally. A possible regime for future investigation would be allowing the quantum states and measurements to be conditioned on the inputs or outputs of past runs, producing a set intermediate between $\qset{1}^{\times N}$ and $\qset{N}$. 

While quantum distributions can violate Bell inequalities, they still obey the no-signalling conditions, preventing faster-than-light communication. In the i.i.d.\ case, the no-signalling conditions can be expressed mathematically as
\begin{equation}
\begin{aligned}
\sum_b P(ab|xy) = \sum_b P(ab|xy') \quad \forall a,x,y,y',\\
\sum_a P(ab|xy) = \sum_a P(ab|x'y) \quad \forall b,y,x,x'.
\end{aligned}
\label{nsconstraints}
\end{equation}

\noindent For any fixed $P(xy)$, these specify a set of linear equality constraints on $\pset{1}$. This hence defines a slice of the polytope $\pset{1}$, which we denote as $\ns{1}$, the no-signalling polytope. The no-signalling set can be easier to study than the quantum set, since it can be characterised by a finite set of vertices, unlike the quantum set. However we note that when studying MDL models, the MDL polytope is not constrained to lie on the no-signalling slice, because the measurement dependence in $P(xy|\lambda)$ can introduce correlations between the inputs. 

A particularly significant point on the no-signalling slice is the Popescu-Rohrlich (PR) box~\cite{popescu94}, defined by
\begin{align}
P_\text{PR}(ab|xy) =  \left\{ 
\begin{array}{lr}
\frac{1}{2} &\text{if } a \oplus b = xy \\
0 &\text{otherwise}
\end{array}
\right. ,
\label{probsPR}
\end{align}

\noindent where $\oplus$ represents addition modulo 2. This distribution satisfies the no-signalling constraints, but cannot be achieved by any quantum models. It has the important property that up to permutations of inputs, outputs and parties, every vertex of $\ns{1}$ in the bipartite 2-input 2-output case is either a vertex of $\lrset{1}$ or a PR box~\cite{barrett05ns}. 

We can also generalise to the product set $\ns{1}^{\times N} \subseteq \pset{N}$, defined as the set of points admitting a decomposition
\begin{align}
P(\vec{a} \vec{b} \vec{x} \vec{y}) = \left( \int d\lambda \, w(\lambda) \prod_{j=1}^N P_{q_j(\lambda)}(a_j b_j | x_j y_j) \right) P(\vec{x} \vec{y}),
\label{nssetn}
\end{align}

\noindent where all the points $q_j(\lambda) \in \pset{1}$ satisfy the i.i.d.\ no-signalling constraints in Eq.~\eqref{nsconstraints}. Since quantum models are no-signalling, $\qset{1}^{\times N}$ is clearly a subset of $\ns{1}^{\times N}$. Intuitively, we would also expect that allowing block-i.i.d.\ quantum models still does not result in apparently-signalling distributions after averaging over the runs, which is to say $\cg{N}(\qset{N}) \subseteq \ns{1}$. We now show that this is indeed the case, at least on the uniform-measurements slice.

\begin{proposition}
\label{prop_ns}
Consider any $q \in \pset{N}$ on the uniform-measurements slice $P(\vec{x}\vec{y}) = 1/(d_X d_Y)^N$. If its conditional probabilities $P(\vec{a}\vec{b}|\vec{x}\vec{y})$ satisfy the constraints
\begin{equation}
\begin{aligned}
\sum_{\vec{b}} P(\vec{a}\vec{b}|\vec{x}\vec{y}) = \sum_{\vec{b}} P(\vec{a}\vec{b}|\vec{x}\vec{y}') \quad \forall \vec{a},\vec{x},\vec{y},\vec{y}',\\
\sum_{\vec{a}} P(\vec{a}\vec{b}|\vec{x}\vec{y}) = \sum_{\vec{a}} P(\vec{a}\vec{b}|\vec{x}'\vec{y}) \quad \forall \vec{b},\vec{y},\vec{x},\vec{x}',
\end{aligned}
\label{nsnconstraints}
\end{equation}

\noindent then the conditional probabilites $P(ab|xy)$ corresponding to $\cg{N}(q) \in \pset{1}$ satisfy the i.i.d.\ no-signalling constraints specified in Eq.~\eqref{nsconstraints}, and thus $\cg{N}(q) \in \ns{1}$ with $P(xy) = 1/(d_X d_Y)$.
\end{proposition}

\begin{corollary*}
If the uniform-measurements constraint $P(\vec{x}\vec{y}) = 1/(d_X d_Y)^N$ is imposed, we have $\cg{N}(\qset{N}) \subseteq \ns{1}$ and $\cg{N}(\ns{1}^{\times N}) = \ns{1}$, with $P(xy) = 1/(d_X d_Y)$.
\end{corollary*}

\begin{proof}
Suppose that the conditions in Eq.~\eqref{nsnconstraints} are fulfilled by the point $q \in \pset{N}$. Given the uniform-measurements constraint $P(\vec{x}\vec{y}) = 1/(d_X d_Y)^N$, these conditions on $P(\vec{a}\vec{b}|\vec{x}\vec{y})$ can be converted directly to the same statements for $P(\vec{a}\vec{b}\vec{x}\vec{y})$ simply by multiplying throughout by $P(\vec{x}\vec{y}) = 1/(d_X d_Y)^N$. We can hence write
\begin{equation}
\begin{aligned}
\sum_{\vec{b}} P_q(\vec{a}\vec{b}\vec{x}\vec{y}) = f(\vec{a},\vec{x}),\\
\sum_{\vec{a}} P_q(\vec{a}\vec{b}\vec{x}\vec{y}) = g(\vec{b},\vec{y}),
\end{aligned}
\end{equation}

\noindent expressing the fact that these sums are independent of $\vec{y}$ and $\vec{x}$ respectively. On this slice, we also have $P(xy) = 1/(d_X d_Y)$ for all $(x,y)$. Hence considering the i.i.d.\ no-signalling condition for Alice (the first equation in Eq.~\eqref{nsconstraints}), we note that for any $(a,x,y)$, we have 
\begin{align}
\sum_b P_{\cg{N}(q)}(ab|xy) &= d_X d_Y \sum_b P_{\cg{N}(q)}(abxy) \nonumber \\
&= \frac{d_X d_Y}{N} \sum_b \sum_{j=1}^N \left( \sum_{(\vec{a},\vec{b},\vec{x},\vec{y}) \in S_j(a,b,x,y) } P_q(\vec{a} \vec{b} \vec{x} \vec{y}) \right) \nonumber \\
&= \frac{d_X d_Y}{N} \sum_{j=1}^N \left( \sum_{(\vec{a},\vec{x},\vec{y}) \in S'_j(a,x,y) } \sum_{\vec{b}} P_q(\vec{a} \vec{b} \vec{x} \vec{y}) \right) \nonumber \\
&= \frac{d_X d_Y}{N} \sum_{j=1}^N \left( \sum_{(\vec{a},\vec{x},\vec{y}) \in S'_j(a,x,y) } f(\vec{a},\vec{x}) \right) \nonumber \\
&= \frac{d_X d_Y}{N} \sum_{j=1}^N \left( d_Y^{N-1} \sum_{(\vec{a},\vec{x}) \in S''_j(a,x) } f(\vec{a},\vec{x}) \right), \quad \text{since the summand is independent of $\vec{y}$} \nonumber \\
&= \frac{d_X d_Y^N}{N} \sum_{j=1}^N \left(\sum_{(\vec{a},\vec{x}) \in S''_j(a,x) } f(\vec{a},\vec{x}) \right),
\end{align}

\noindent where $S'_j(a,x,y)$ is defined as the set of tuples $(\vec{a},\vec{x},\vec{y})$ such that $(a_j,x_j,y_j) = (a,x,y)$, and similarly for $S''_j(x,y)$, analogous to the definition of $S_j(a,b,x,y)$ for the coarse-graining function. From the final expression, we see that $\sum_b P_{\cg{N}(q)}(ab|xy)$ is independent of $y$, and thus the i.i.d.\ no-signalling condition for Alice is fulfilled. Applying the same argument to Bob, we conclude that indeed, $\cg{N}(q) \in \ns{1}$ with $P(xy) = 1/(d_X d_Y)$.

As for the corollary, the statement $\cg{N}(\qset{N}) \subseteq \ns{1}$ follows immediately by noting that even for the quantum points described in Eq.~\eqref{qsetncoherent}, the constraints in Eq.~\eqref{nsnconstraints} are still satisfied, as can be seen by treating it as an i.i.d.\ scenario with a larger alphabet. Regarding $\cg{N}(\ns{1}^{\times N})$, we similarly have $\cg{N}(\ns{1}^{\times N}) \subseteq \ns{1}$, since it can be shown that any point in $\ns{1}^{\times N}$ obeys the constraints of Eq.~\eqref{nsnconstraints}. We note also that Proposition~\ref{prop_cg} implies $\ns{1} \subseteq \cg{N}(\ns{1}^{\times N})$ on the uniform-measurements slice, since it shows that any point in $\ns{1}$ with $P(xy) = 1/(d_X d_Y)$ has a pre-image in $\ns{1}^{\times N}$ with $P(xy) = 1/(d_X d_Y)^N$ under the coarse-graining function (simply by repetition). Therefore, we can conclude that $\cg{N}(\ns{1}^{\times N}) = \ns{1}$ under the uniform-measurements condition.
\end{proof}

\section{Vertices of the MDL$_N$ polytope}

We now turn to the issue of characterising the MDL$_N$ set, as defined in Eq.~\ref{mdlndef} and subject to the constraints $L_{} \leq P(\vec{x} \vec{y} | \lambda) \leq H_{}$. Regarding these constraints, we note that any $L_{}>0$ implicitly imposes an upper bound $P(xy|\lambda) \leq 1-((d_X d_Y)^N - 1)L_{}$ by normalisation of $P(xy|\lambda)$. Similarly, any $H_{} < 1/((d_X d_Y)^N - 1)$ implies a nonzero lower bound $P(xy|\lambda) \geq 1-((d_X d_Y)^N - 1)H_{}$. In particular, for the 2-input 2-output i.i.d.\ case, this shows that any $h < 1/3$ imposes a nonzero lower bound, even if $l$ is left unspecified. 

In cases of potential ambiguity, we shall refer to the general MDL$_N$ models in Eq.~\eqref{mdlndef} as dependent-runs models, and those constrained to independent local strategies (Eq.~\eqref{indrundef}) as independent-runs models. Their corresponding sets are denoted $\mdlset{N}$ and $\mdlset{N}'$ respectively. There may be other MDL$_N$ sets of interest, such as models where the outputs depend on the inputs of all past runs but not future runs. However, in this work we shall only consider the dependent-runs and independent-runs models. 

As claimed in the main text, the set of probability distributions admitting dependent-runs or independent-runs MDL$_N$ models forms a polytope in \pset{N}. When subsequently taking a slice of $\pset{N}$ by specifying $P(\vec{x}\vec{y})$, the values chosen for $P(\vec{x}\vec{y})$ must be compatible with the constraints $(L,H)$, in order for the MDL$_N$ set to have non-empty intersection with this slice.  For $\mdlset{N}$, its vertices are precisely the set of points of the form
\begin{align}
P(\vec{a} \vec{b} \vec{x} \vec{y}) = P(\vec{a} | \vec{x}) P(\vec{b} | \vec{y}) P(\vec{x} \vec{y}),
\label{depverts}
\end{align}

\noindent where $P(\vec{a} | \vec{x})$ and $P(\vec{b} | \vec{y})$ are all equal to either 0 or 1, and the values of $P(\vec{x} \vec{y})$ are extremal in the sense that all but at most one of them are either equal to $L_{}$ or $H_{}$. Such an assignment of values for $P(\vec{a} | \vec{x})$ and $P(\vec{b} | \vec{y})$ is referred to as a local deterministic strategy. Similarly, the vertices of $\mdlset{N}'$ are the set of points of the form 
\begin{align}
P(\vec{a} \vec{b} \vec{x} \vec{y}) = \left(\prod_{j=1}^N P(a_j | x_j) P(b_j | y_j)\right) P(\vec{x} \vec{y}),
\label{indverts}
\end{align}

\noindent where $P(a_j|x_j)$ and $P(b_j|y_j)$ are all equal to either 0 or 1, and the values of $P(\vec{x} \vec{y})$ are extremal as described above.\\

To justify this claim, we note that for dependent-runs models, the MDL$_N$ scenario for $(d_X,d_Y,d_A,d_B)$ is mathematically equivalent to the MDL$_1$ scenario for $(d_X^N,d_Y^N,d_A^N,d_B^N)$. Hence the proof in Refs.~\cite{putz14,putz15} for the i.i.d.\ MDL set with arbitrary finite inputs and outputs carries over directly to this case, and it shows that the dependent-runs MDL$_N$ set is a polytope with vertices of the form in Eq.~\eqref{depverts}. 

For independent-runs models, we instead use an intermediate theorem from Refs.~\cite{putz14,putz15}, that if the conditional output probabilities $P(\vec{a}\vec{b}|\vec{x}\vec{y}\lambda)$ and the input probabilities $P(\vec{x}\vec{y}|\lambda)$ are both drawn from polytopes, then combining them in the manner of Eq.~\ref{mdlndef} produces a polytope. Since the independent-runs restriction only affects $P(\vec{a}\vec{b}|\vec{x}\vec{y}\lambda)$, it suffices to show that these conditional output probabilities still form a polytope under the independent-runs condition. 

By noting that the probabilities $P(\vec{a}\vec{b}|\vec{x}\vec{y})$ that admit an independent-runs decomposition are isomorphic to those for a $2N$-party local realistic model where $N$ parties have $d_X$ inputs and $N$ parties have $d_Y$ outputs, we see that they indeed form a polytope, with vertices given by local deterministic strategies. The theorem from Refs.~\cite{putz14,putz15} then shows that the independent-runs MDL$_N$ set is a polytope, and that its vertices are given by combinations of local deterministic strategies with an extremal assignment of values to $P(\vec{x}\vec{y})$, as described in Eq.~\eqref{indverts}. \\

We see from this that the number of vertices of the MDL$_N$ polytope equals the product of the number of local deterministic strategies with the number of possibilities for an extremal assignment of values to $P(\vec{x} \vec{y})$. In the 2-input 2-output case, there are $((2^N) ^\wedge (2^N))^2 = 2^{(2N 2^N)}$ local deterministic strategies for dependent-runs models, or $4^{2N}$ local deterministic strategies for independent-runs models. As for the number of possible extremal assignments for $P(\vec{x} \vec{y})$, this depends on the values of $L_{}$ and $H_{}$. By directly applying the proof given in Refs.~\cite{putz14,putz15} for the i.i.d.\ case with $2^N$ inputs and outputs, we note that up to permutation, the unique extremal assignment of values to the $2^{2N}$ terms $P(\vec{x} \vec{y})$ is to have $m = \floor{\frac{1-2^{2N} L_{}}{H_{}-L_{}}}$ of them equal to $H_{}$, $2^{2N}-1-m$ of them equal to $L_{}$, and the last chosen to satisfy normalisation. The number of permutations is hence given by the multinomial coefficient 
\begin{align}
\multinom{2^{2N}}{m, (2^{2N}-m-1), 1} = \frac{2^{2N}!}{m! (2^{2N}-m-1)!}.
\label{numperms}
\end{align}

\noindent A special case is when the values of $(L_{}, H_{})$ are such that $\frac{1-2^{2N} L_{}}{H_{}-L_{}}$ is already an integer, in which case all of $P(\vec{x} \vec{y})$ can be set equal to either $L_{}$ or $H_{}$. In that case, letting the number of terms equal to $H_{}$ be $m = \frac{1-2^{2N} L_{}}{H_{}-L_{}}$, the number of possible permutations is the binomial coefficient
\begin{align}
\multinom{2^{2N}}{m, (2^{2N}-m) } = \frac{2^{2N}!}{m! (2^{2N}-m)!}.
\label{numpermscrit}
\end{align}

\noindent From these expressions, we see that for the special values of $(L_{}, H_{})$ where all of $P(\vec{x} \vec{y})$ can be set equal to either $L_{}$ or $H_{}$, the number of vertices of the MDL$_N$ polytope tends to be smaller. However, we see that in general, the number of vertices is large enough that it would be intractable for anything more than small values of $N$. 

\section{The i.i.d.\ MDL inequality}

The MDL inequality shown in Eq.~5 of~\cite{putz14} can be written in the form
\begin{align}
l P(0000) - h(P(0101) + P(1010) + P(0011)) \leq B_N,
\label{putzineq}
\end{align}

\noindent where $B_N$ is the largest value on the left-hand side attainable by MDL$_N$ models. We shall denote the left-hand side of the expression as $\expval{\putz{l}{h}}$, which is consistent with an interpretation in the quantum case as the expectation value of a Bell operator $\putz{l}{h}$. 

For the i.i.d.\ case, we simply have $B_1 = 0$ for any $h \in \left[1/4, 1/3\right)$~\cite{putz14}. In that case, the inequality can be violated by any quantum probability distribution that demonstrates a Hardy-type paradox~\cite{hardy92,hardy93,rabelo12}, where $P(01|01) = P(10|10) = P(00|11) = 0$ but $P(00|00)>0$. This gives $\expval{\putz{l}{h}} = l P(00|00) P_{XY}(00)$, thereby violating the inequality for any $l > 0$. The state and measurements described in Ref.~\cite{putz14} exhibit this Hardy-type paradox with $P(00|00) = 1/12$, but this is not the maximum value of $P(00|00)$ that can be achieved by quantum models under the conditions $P(01|01) = P(10|10) = P(00|11) = 0$. Instead, the maximum value is $P(00|00) = (5\sqrt{5}-11)/2$; it is attained by the state
\begin{align}
\alpha \left(\ket{01} + \ket{10}\right) + \sqrt{1-2\alpha^2} \ket{11},
\label{hardystate}
\end{align}

\noindent measured using projective measurements on the states $\ket{a_0} = \ket{b_0} = \sqrt{1-2\alpha^2} \ket{0} - \alpha \ket{1}$, $\ket{a_1} = \ket{b_1} = \ket{0}$ (up to normalisation), with $\alpha = \sqrt{(3-\sqrt{5})/2}$~\cite{rabelo12}. It turns out that if we do not impose the condition $P(01|01) = P(10|10) = P(00|11) = 0$, slightly higher quantum values of $\expval{\putz{l}{h}}$ can be achieved, but the analysis becomes more complicated because we need to specify all the values $P(xy)$, rather than just $P_{XY}(00)$. Hence in subsequent discussion, we only consider the quantum point $q_\text{Hardy} \in \pset{1}$ defined by the state and measurements in Eq.~\eqref{hardystate}.

It is important to note that different values of $(l,h)$ do not affect the inequality in Eq.~(\ref{putzineq}) by changing the bound on the right-hand side, but rather by changing the coefficients on the left-hand side. This has some implications for numerical analysis of the block-i.i.d.\ case, and thus we shall now take some care in precisely specifying the scenario under consideration. We analyse a scenario where experimenters measure the value of $\expval{\putz{l}{h}}$ given by $q_\text{Hardy}$, with the objective of ruling out all i.i.d.\ MDL models subject to the constraint $P(xy|\lambda) \leq h$ for a specific $h < 1/3$. Under this experimental scenario, this value of $h$ implicitly creates a lower bound $l = 1-3h$, which would be the value of $l$ used by the experimenters in $\expval{\putz{l}{h}}$ to obtain a nonzero quantum violation. We now wish to investigate whether the experimenters' results could instead have been produced by a block-$N$-i.i.d.\ model with the same average min-entropy per run, as described in Eq.~\eqref{minentropy}. Since the probabilities in Eq.~\eqref{putzineq} are of the form $P(abxy)$, we will work in the probability space $\pset{1}$, using the coarse-graining function as necessary. Finally, for this section only, we do not restrict ourselves \textit{a priori} to some fixed set of values for $P(\vec{x} \vec{y})$.

\begin{figure}
\centering
\subfloat[Dependent-runs model]{
	\includegraphics[width=0.48\textwidth]{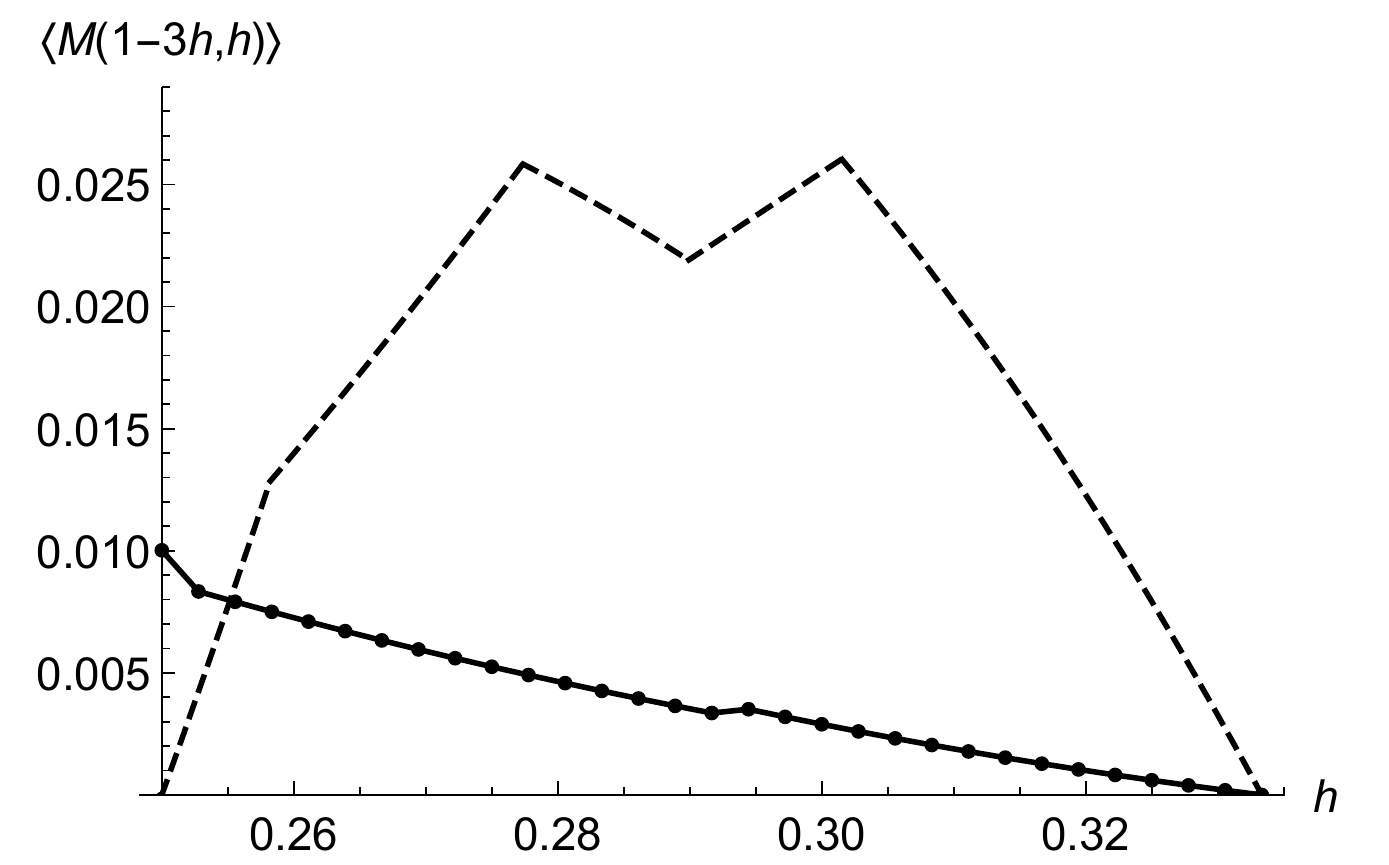}
}
\subfloat[Independent-runs model]{
	\includegraphics[width=0.48\textwidth]{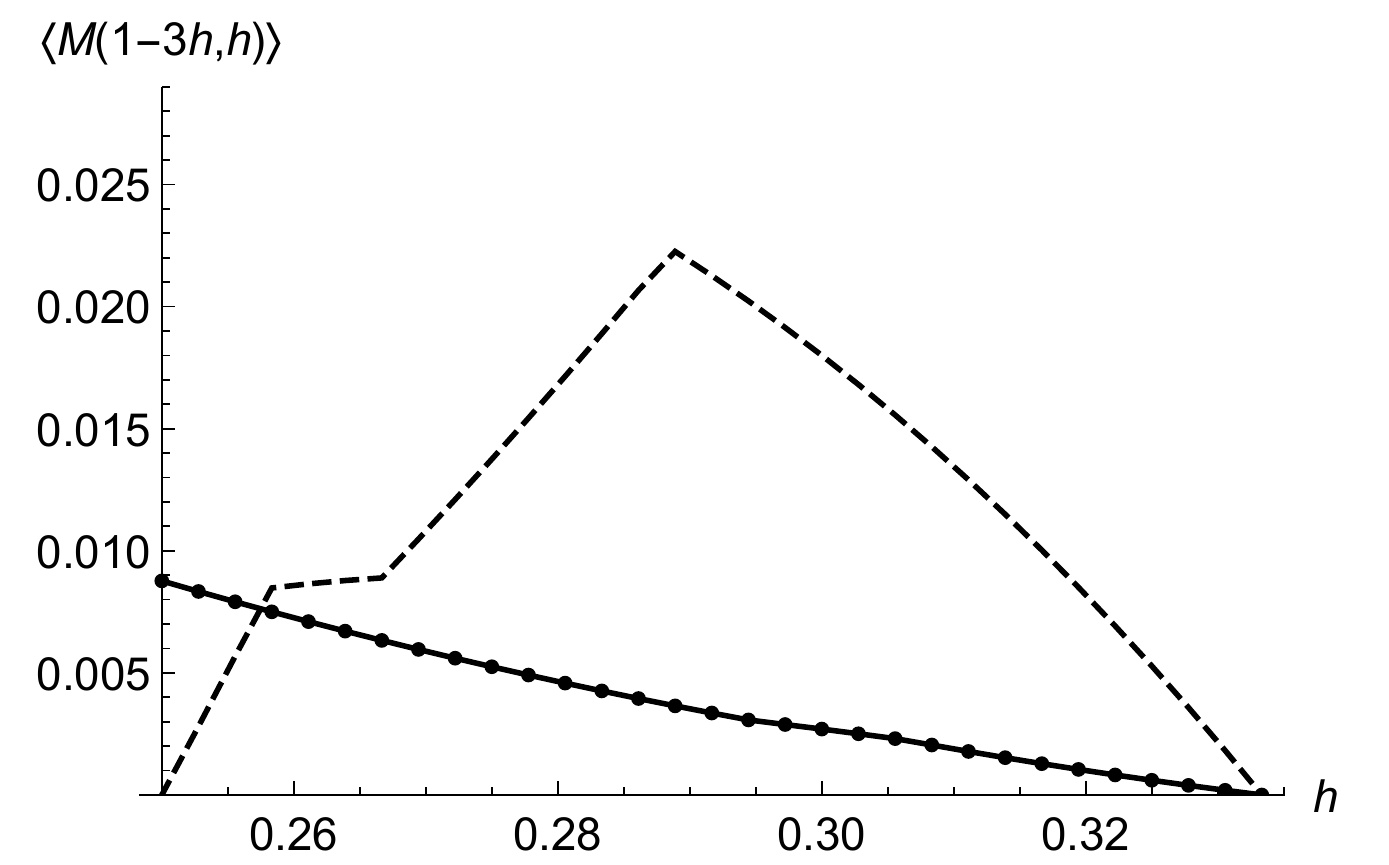}
}
\caption{Maximum value of $\expval{\putz{1-3h}{h}}$ for MDL$_2$ models in the dependent-runs and independent-runs cases, as shown by the dashed curves. The solid curves represent the value achieved by the quantum point $q_\text{Hardy}$. It can be seen that in both cases, the MDL$_2$ value exceeds the quantum value at a low threshold value of $h$, indicating that the violation of the inequality in Eq.~\eqref{putzineq} by this quantum state is not sufficient to rule out MDL$_2$ models with values of $h$ above this threshold.}
\label{fig_putzgraphs}
\end{figure}

For MDL$_N$ models with $N>1$, the value of $B_N$ would be larger than or equal to that for MDL$_1$ models. Given a fixed $h$, this maximum value is achieved at one of the vertices of $\mdlset{N}$ or $\mdlset{N}'$, which have the form shown in Eq.~\eqref{depverts} or Eq.~\eqref{indverts}. We could hence obtain $B_N$ by computing the left-hand side of Eq.~\eqref{putzineq} with respect to all the vertices, then taking the largest value. However, there is a more efficient method, by noting that each vertex is obtained by combining a local deterministic strategy with an extremal assignment of values for $P(\vec{x}\vec{y})$. This implies that for any given local deterministic strategy, maximising the value of $\expval{\putz{1-3h}{h}}$ is a linear program over the variables $P(\vec{x}\vec{y})$, subject to the constraints $P(\vec{x}\vec{y}) \leq h^N$. This can be efficiently solved for each local deterministic strategy, and the largest value over all local deterministic strategies is then the value of $B_N$. \\

The results for the $N=2$ case are shown in Fig.~\ref{fig_putzgraphs}. From the graphs, we see that the maximum value of $\expval{\putz{1-3h}{h}}$ that can be achieved by MDL$_2$ models already exceeds the quantum value at a low threshold value of $h$, subsequently denoted as $h_0$. In the dependent-runs case, we have $h_0 \approx 0.255$, while in the independent-runs case we have $h_0 \approx 0.257$. This hence shows that the inequality in Eq.~\eqref{putzineq}, which admits a quantum violation for any $h \in \left[1/4, 1/3\right)$ in the i.i.d.\ case, already becomes substantially less robust in the block-2-i.i.d.\ scenario.

For the graphs in Fig.~\ref{fig_putzgraphs}, the quantum value shown is $\frac{1-3h}{12} P_{XY}(00)$, which requires us to specify a value for $P_{XY}(00)$. The value used is that corresponding to the MDL$_2$ strategy which yields the highest value of $\expval{\putz{1-3h}{h}}$, or in the cases where there were multiple such strategies, the one with the largest value of $P_{XY}(00)$ was chosen. For all points in the graphs, this value of $P_{XY}(00)$ was larger than $1/4$, highlighting the fact that they do not satisfy the uniform-measurements constraints. \\

Finding $B_N$ subject to the uniform-measurements condition is more difficult. To do so, we chose to generate all vertices of the MDL$_N$ polytope, and treat $\expval{\putz{1-3h}{h}}$ as a weighted combination of the probabilities $P(\vec{a}\vec{b}\vec{x}\vec{y})$ at the vertices. Maximising $\expval{\putz{1-3h}{h}}$ under the uniform-measurements constraint is then a linear program over the weights. The results are shown in Fig.~\ref{fig_putzuniform}. For this case, we only obtained results for the independent-runs model, because the number of vertices for the dependent-runs model was beyond the scope of our computational resources. Since the uniform-measurements constraint has been imposed, we now use  $P_{XY}(00) = 1/4$ throughout in the quantum value. We see that even with this constraint, the maximum value of $\expval{\putz{1-3h}{h}}$ for MDL$_2$ models only decreases slightly, and still exceeds the quantum value at a low threshold $h_0$.

A feature that can be observed from the graphs in Figs.~\ref{fig_putzgraphs} and \ref{fig_putzuniform} is that they appear to have a piecewise structure, with gradient discontinuities at particular values of $h$. In almost all cases, these gradient discontinuities occur at the values $h^2 = 1/15, 1/14, ..., 1/10$. These are the critical values of $h$ at which it becomes possible to set one more input probability $P(\vec{x} \vec{y}|\lambda)$ equal to zero for each $\lambda$-strategy. Using this idea, we were able to derive the explicit $\lambda$-strategies used in each interval to maximise $\expval{\putz{1-3h}{h}}$, and hence obtain piecewise closed-form expressions for the graphs. These were used to plot the dashed curves, and hence the data points have been omitted from such curves.\\

For the independent-runs model, by studying the MDL$_2$ strategies that maximised the value of $\expval{\putz{1-3h}{h}}$, we were also able to obtain an extension to the MDL$_N$ case. This is analogous to the approach used by Pope and Kay~\cite{pope13}, but developed for this i.i.d.\ MDL inequality rather than the CHSH inequality. Specifically, for any integer $k \in [0,N]$, consider the local deterministic strategy of Alice and Bob always having output 1 in the first $N-k$ runs, and always having output 0 in the last $k$ runs. It can be shown that if we set $P(\vec{x} \vec{y}) = 0$ for all $(\vec{x}, \vec{y})$ where $x_j y_j = 11$ in any of the last $k$ runs, then we obtain $P(0101) = P(1010) = P(0011) = 0$. This MDL$_N$ strategy requires the value of $h$ to be greater than or equal to a threshold value
\begin{align}
&h_0^N = \frac{1}{3^k 4^{N-k}} \implies h_0 = \frac{1}{4} \left(\frac{4}{3}\right)^\frac{k}{N},
\label{h0mdln}
\end{align}

\noindent which in turn yields $P(0000) = k/(3N)$. In summary, this implies that for any $(N,k)$, an MDL$_N$ model with $h = (1/4)\left(4/3\right)^{k/N}$ can always achieve a value of at least $\expval{\putz{1-3h}{h}} = (1-3h)k/({3N})$, while still satisfying $P(0101) = P(1010) = P(0011) = 0$. For sufficiently large $N$ that the discrete nature of $(N,k)$ can be neglected, this can be viewed as a locus of points in a plot of $\expval{\putz{1-3h}{h}}$ against $h$, that gives a lower bound on the value of $\expval{\putz{1-3h}{h}}$ achievable by MDL$_N$ models. In Fig.~\ref{fig_putzuniform}, we have plotted this for comparison to the MDL$_2$ result. As expected, it is higher than the MDL$_2$ value. \\

\begin{figure}
\centering
\includegraphics[width=0.9\textwidth]{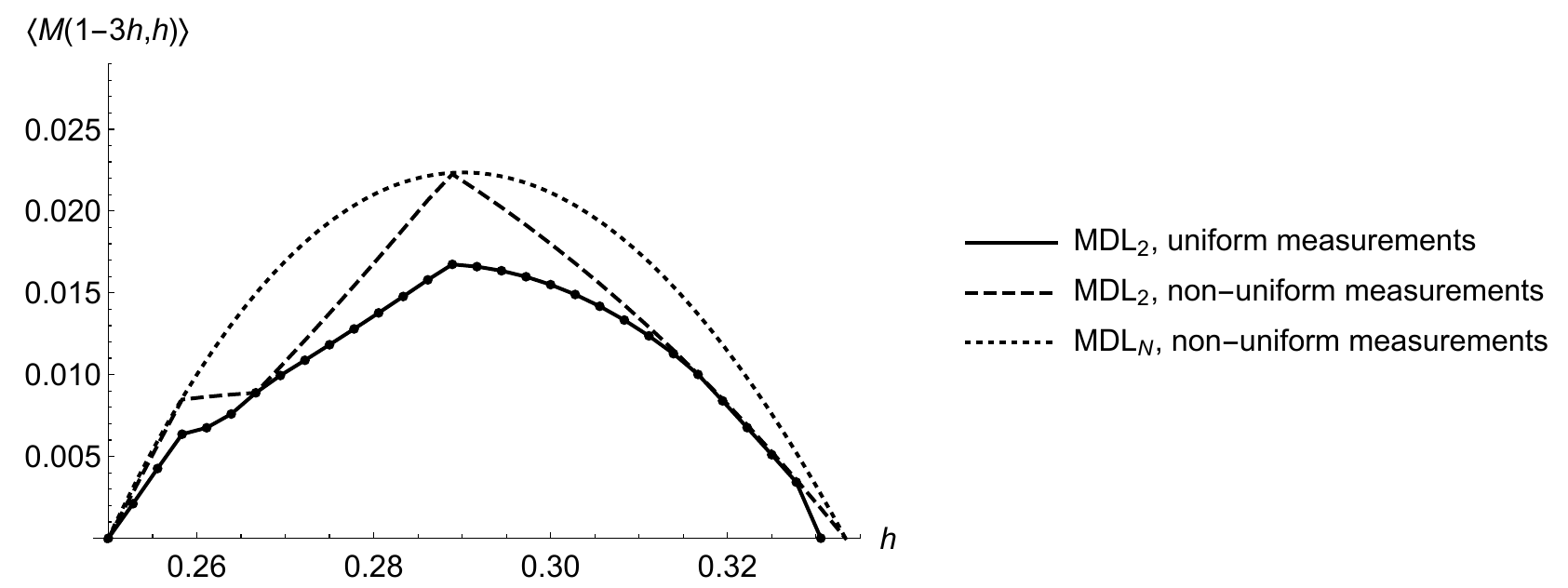}
\caption{Value of $\expval{\putz{1-3h}{h}}$ for various independent-runs MDL$_N$ models. The solid curve shows the maximum value that can be achieved by MDL$_2$ models under the uniform-measurements constraint, which as expected, is lower than the maximum value attainable without the constraint, shown by the dashed curve. The dotted curve shows a locus of values of $\expval{\putz{1-3h}{h}}$ that we have shown to be attainable by MDL$_N$ models without the uniform-measurements constraint when $N$ is large. This is hence a lower bound on the maximum value achievable by such MDL$_N$ models, and as expected, it is already higher than the maximum value for the MDL$_2$ model.}
\label{fig_putzuniform}
\end{figure}

Thus far, we have only considered the value $\expval{\putz{1-3h}{h}}$ as a single quantity. However, $q_\text{Hardy}$ does not only give the value of $\expval{\putz{1-3h}{h}}$, but also more specifically the four probabilities $P(0000)$, $P(0101)$, $P(1010)$, $P(0011)$. Using some techniques described in the next section, we find that the set of values $P(0000) > 0$, $P(0101) = P(1010) = P(0011) = 0$ is achievable if and only if $h^2 \geq \frac{1}{14}$ for dependent-runs models, or $h^2 \geq \frac{1}{12}$ for independent-runs models. This indicates the existence of a different MDL$_2$ inequality that is violated by $q_\text{Hardy}$ for values of $h$ below these thresholds, though our methods do not allow us to directly find this inequality.

We could potentially carry on studying the behaviour of $\expval{\putz{1-3h}{h}}$ in greater detail. However, our current results already suffice to show that this i.i.d.\ MDL inequality has low robustness in the block-2-i.i.d.\ scenario alone, and it can only become even less robust for larger values of $N$. We hence instead turn to the question of characterising the overall structure of the MDL$_N$ polytope, with the aim of possibly finding more robust MDL inequalities. 

\section{Methods for studying the MDL$_N$ polytope}

From this point forward, we consider only the $(d_X,d_Y,d_A,d_B)=(2,2,2,2)$ case. The number of vertices of the MDL$_2$ polytope is then already up to \sci{2.05}{7} for independent-runs models, which we could still at least generate computationally, but it can be up to \sci{5.24}{9} for dependent-runs models, which we found to be intractable. Although the number of vertices does become smaller for values of $h$ closer to $1/4$, it is helpful to introduce some simplifying techniques in order to study the polytope in general. These techniques are described below with respect to $\mdlset{N}$, but are also valid for $\mdlset{N}'$. \\

\textbf{Compatible and incompatible vertices ---} Given the list of vertices of $\mdlset{N}$, determining whether some point $q$ is in $\mdlset{N}$ can be cast as a linear program. We can simplify this task for a specific $q$ using the notion of \textit{compatible} and \textit{incompatible} vertices. Suppose that for the point $q$, some of the probabilities $P_q(\vec{a} \vec{b} \vec{x} \vec{y})$ are zero. If $q$ could be written as a convex combination of vertices of $\mdlset{N}$, then this convex combination cannot include any vertices with nonzero values of these probabilities. We shall say that such vertices are \textit{incompatible with $q$}. Therefore, we can remove all incompatible vertices before running the linear program, without affecting the conclusion of whether $q$ is in $\mdlset{N}$.\\

\textbf{$\mathbf{k}$-mismatch strategies ---} A further improvement on the concept of incompatible vertices can be made, by recalling that the vertices of $\mdlset{N}$ are given by combining local deterministic strategies with extremal values of the input probabilities $P(\vec{x} \vec{y})$. Consider any specific local deterministic strategy, with conditional output probabilities $P_\text{det}(\vec{a} \vec{b} | \vec{x} \vec{y})$. For it to produce a vertex compatible with $q$ after multiplying by $P(\vec{x} \vec{y})$, it must satisfy $P_\text{det}(\vec{a} \vec{b} | \vec{x} \vec{y}) P(\vec{x} \vec{y}) = 0$ for all $(\vec{a},\vec{b},\vec{x},\vec{y})$ where $P_q(\vec{a} \vec{b} \vec{x} \vec{y}) = 0$. Hence by considering cases where $P_q(\vec{a}\vec{b}\vec{x}\vec{y}) = 0$ but $P_\text{det}(\vec{a} \vec{b} | \vec{x} \vec{y}) \neq 0 $, we can deduce how many terms $P(\vec{x} \vec{y})$ need to be set equal to 0 in order to produce a vertex compatible with $q$ from this local deterministic strategy. Supposing that there are $k$ such terms, we shall refer to this local deterministic strategy as a \textit{$k$-mismatch strategy with respect to $q$}. 

This notion is significant because under the min-entropy constraint of Eq.~\eqref{minentropy}, the normalisation requirement enforces that at most $k_\text{max} = 2^{2N} - \ceil{1/h^N}$ of the terms $ P(\vec{x} \vec{y} | \lambda)$ can be set equal to 0 for a given $\lambda$. Hence for any given point $q \in \pset{N}$ and value of $h$, if a local deterministic strategy is a $k$-mismatch strategy with respect to $q$ such that $k > k_\text{max}$, then it cannot produce any vertices of $\mdlset{N}$ that are compatible with $q$. This allows us to entirely omit such local deterministic strategies when generating vertices of $\mdlset{N}$, if we are only considering whether the specific point $q$ is in $\mdlset{N}$. In addition, when generating vertices of $\mdlset{N}$ from the remaining local deterministic strategies, we can omit those corresponding to any assignments of $P(\vec{x} \vec{y})$ that do not produce a compatible vertex. \\

These methods can only be applied for points $q$ where some of the probabilities $P_q(\vec{a}\vec{b}\vec{x}\vec{y})$ are zero, but can be very effective in some cases. In addition, they can sometimes allow an immediate conclusion without needing to run a linear program. Again considering some point $q \in \pset{N}$, we can list all local deterministic strategies and label each as some $k$-mismatch strategy with respect to $q$. If we find that the minimum value of $k$ amongst all these strategies is some $k_0 > 0$, then we can already conclude that $q \not\in \mdlset{N}$ for any $h^N < 1/(2^{2N} - k_0)$, because there cannot be any vertices compatible with $q$ for such values of $h$. This hence gives a simple lower bound for the value of $h$ beyond which $q$ becomes enclosed by $\mdlset{N}$, though this bound may not be tight.

The concept of $k$-mismatch strategies can also be applied to the coarse-grained probabilities, although the reasoning is more complex. Specifically, for some local deterministic strategy $P_\text{det}(\vec{a} \vec{b} | \vec{x} \vec{y})$ and point $q \in \pset{N}$, $k$ is then the number of terms $P(\vec{x} \vec{y})$ that need to be set equal to 0 in order to produce a vertex compatible with $\cg{N}(q)$. To determine this number, we need to consider the cases where $P_{\cg{N}(q)}(abxy) = 0$ but $P_\text{det}(\vec{a} \vec{b} | \vec{x} \vec{y}) \neq 0$ for at least one of the terms in the summation used to compute such $P_{\cg{N}(q)}(abxy)$. 

\section{Properties of the MDL$_N$ polytope}

From this point forward, we only consider the uniform-measurements slice $P(\vec{x} \vec{y}) = 1/4^N$. Using the techniques described above, we studied each of the cases described in Table~\ref{table_summary}, focusing on MDL$_2$ models which were still computationally tractable. We first describe the results for the dependent-runs models, followed by independent-runs models.\\

\textbf{Dependent-runs models (compared to the no-signalling set) ---} For the top row of Table~\ref{table_summary}, we obtained results by considering the set of points in $\pset{2}$ of the form
\begin{align}
P(\vec{a} \vec{b} \vec{x} \vec{y}) = \frac{1}{16} P_{1}(a_1 b_1 | x_1 y_1) P_{2}(a_2 b_2 | x_2 y_2),
\label{ns2verts}
\end{align}

\noindent where each of $P_{1}(ab|xy),P_{2}(ab|xy)$ is either a local deterministic strategy or the PR box distribution. Applying the simplifying techniques described above, we found that all such points are contained in $\mdlset{2}$ when $h^2 \geq 1/10$. Since every point in $\ns{1}^{\times 2}$ is a convex combination of such points, this implies that $\ns{1}^{\times 2} \subseteq \mdlset{2}$ for all $h^2 \geq 1/10$. In addition, for the point $q_\text{PR2}$ where both $P_{1}(ab|xy)$ and $P_{2}(ab|xy)$ are PR boxes, every local deterministic strategy is at least a 6-mismatch strategy with respect to $q_\text{PR2}$, and hence we also deduce that $q_\text{PR2} \not\in \mdlset{2}$ for any $h^2 < 1/10$. (The decomposition of $q_\text{PR2}$ as a convex combination of the vertices of $\mdlset{2}$ at $h^2 = 1/10$ is shown at the end of this Supplemental Material.) We thus conclude that $\ns{1}^{\times 2} \subseteq \mdlset{2}$ if and only if $h^2 \geq 1/10$. 

This immediately implies that for the coarse-grained probabilities, we have $\cg{2}(\ns{1}^{\times 2}) \subseteq \cg{2}(\mdlset{2})$ for all $h^2 \geq 1/10$; also, we recall that $\cg{2}(\ns{1}^{\times 2}) = \ns{1}$ by the corollary of Proposition~\ref{prop_ns}. In addition, we found that even for the coarse-grained probabilities, every local deterministic strategy is at least a 6-mismatch strategy with respect to the PR box with uniform measurements. This leads to the conclusion that $\ns{1} \subseteq \cg{2}(\mdlset{2})$ if and only if $h^2 \geq 1/10$. \\

\textbf{Dependent-runs models (compared to the quantum set) ---} The above results were derived for the no-signalling sets, and allow us to conclude that $\qset{1}^{\times 2} \subseteq \mdlset{2}$ and $\cg{N}(\qset{N}) \subseteq \cg{2}(\mdlset{2})$ for any $h^2 \geq 1/10$. However, since the quantum sets are usually proper subsets of the no-signalling sets, it is possible that these quantum sets are enclosed by the MDL$_2$ polytopes at some smaller threshold value $h_0$. We can provide a lower bound on this threshold value by considering the Hardy state again. Specifically, we first consider the point $q_\text{Hardy2} \in \pset{2}$ obtained by simply repeating $q_\text{Hardy}$ (with uniform measurements $P(xy) = 1/4$),
\begin{align}
P_{q_\text{Hardy2}}(\vec{a} \vec{b} \vec{x} \vec{y}) = P_{q_\text{Hardy}}(a_1 b_1 x_1 y_1) P_{q_\text{Hardy}}(a_2 b_2 x_2 y_2).
\end{align}

\noindent The reasoning previously used for the PR box does not immediately generalise to this point, because we found that there exist some 0-mismatch strategies with respect to $q_\text{Hardy2}$. However, a modification to the argument allows us to obtain some results. Namely, $q_\text{Hardy2}$ has $P(\vec{0}\vec{0}\vec{0}\vec{0}) > 0$, and hence for it to be written as a convex combination of vertices of $\mdlset{2}$, there must be at least one vertex with $P(\vec{0}\vec{0}\vec{0}\vec{0}) > 0$ in the convex combination, which in turn must be generated by a local deterministic strategy with $P(\vec{0}\vec{0}|\vec{0}\vec{0}) > 0$. However, we found that all such local deterministic strategies are at least 3-mismatch strategies with respect to $q_\text{Hardy2}$. We can thus conclude that for all $h^2 < 1/13$, we have $q_\text{Hardy2} \not\in \mdlset{2}$ and thus $\qset{1}^{\times 2} \not\subseteq \mdlset{2}$. 

Similarly, we note that $q_\text{Hardy} \in \pset{1}$ has $P(0000) > 0$, but every local deterministic strategy such that $P(0000) > 0$ after coarse-graining is at least a 2-mismatch strategy with respect to $q_\text{Hardy}$. Therefore, for all $h^2 < 1/14$, we have $q_\text{Hardy} \not\in \cg{2}(\mdlset{2})$ and thus $\qset{1} \not\subseteq \cg{2}(\mdlset{2})$. In principle, this implies the existence of MDL$_2$ inequalities in terms of $P(\vec{a} \vec{b} \vec{x} \vec{y})$ or $P(abxy)$ that are violated by $q_\text{Hardy2}$ and $q_\text{Hardy}$ respectively for these ranges of $h$. However, as we did not actually generate the vertices of $\mdlset{2}$ to obtain this result, we do not have the explicit form of these inequalities. We also applied the same argument to all other probabilities $P(\vec{a} \vec{b} \vec{x} \vec{y})$ or $P(abxy)$ that were nonzero for these points, but $P(\vec{0}\vec{0}\vec{0}\vec{0})$ and $P(0000)$ were the ones which gave the best bounds.\\

\textbf{Independent-runs models ---} For the bottom row of Table~\ref{table_summary}, we applied the above argument again for $q_\text{Hardy2}$, but restricted to independent-runs local deterministic strategies. This time, we obtained the result that all local deterministic strategies with $P(\vec{0}\vec{0}|\vec{0}\vec{0}) > 0$ are at least 7-mismatch strategies with respect to $q_\text{Hardy2}$. Therefore, we were able to conclude that $\qset{1}^{\times 2} \not\subseteq \mdlset{2}'$ for all $h < 1/3$, showing that the i.i.d.\ result does generalise to this situation at least. An interesting finding was that with respect to $q_\text{Hardy2}$, all independent-runs local deterministic strategies are $k$-mismatch strategies with $k \in \{0, 4, 7, 10, 12\}$, unlike dependent-runs models where all values $k \leq 12$ were obtained. This may suggest some structure in the independent-runs models that could be exploited for further study. 

In the coarse-grained space, however, we instead found that with respect to $q_\text{Hardy}$, every local deterministic strategy such that $P(0000) > 0$ after coarse-graining is at least a 4-mismatch strategy, so we can only conclude that $q_\text{Hardy} \not\in \cg{2}(\mdlset{2}')$ for $h^2 < 1/12$. We were also able to explicitly generate all vertices of $\mdlset{2}'$, and found that $q_\text{Hardy}$ could be written as a convex combination of these vertices when $h^2 = 1/12$. Therefore, $q_\text{Hardy} \not\in \cg{2}(\mdlset{2}')$ if and only if $h^2 < 1/12$. This implies that $\qset{1} \not\subseteq \cg{2}(\mdlset{2}')$ for all $h^2 < 1/12$, but it is unclear whether it still holds for any larger values of $h$. We applied various methods to search for quantum points outside of $\cg{2}(\mdlset{2}')$ for $h^2 \geq 1/12$, but were unable to find any such points. (In this task, we made use of the fact that for maximising the quantum value in a 2-input 2-output Bell test, it suffices to consider only qubit states~\cite{masanes05}.) While this does not constitute a proof that $\cg{2}(\qset{1}^{\times 2}) \subseteq \cg{2}(\mdlset{2}')$ when $h^2 \geq 1/12$, it does provide some numerical evidence in favour of that possibility. 

We note that our previous approach of studying the no-signalling set instead of the quantum set in the dependent-runs case would not be effective in the independent-runs case. This is because the results of Pope and Kay~\cite{pope13} indicate that for any $N$, the PR box distribution remains outside of $\cg{N}(\mdlset{N}')$ for all $h < 1/3$, and thus $\ns{1} \not\subseteq \cg{N}(\mdlset{N}')$ in this range. In turn, this implies that $\ns{1}^{\times N} \not\subseteq \mdlset{N}'$ for all $h < 1/3$, because $\cg{N}(\ns{1}^{\times N}) = \ns{1}$. Hence whether we consider the full or coarse-grained probabilities, the independent-runs MDL$_N$ polytope does not enclose the no-signalling set in the whole range $h < 1/3$, and so we cannot use it to draw new conclusions about the quantum set.

\clearpage

\section{Decomposition of the point $q_\text{PR2}$}

\begin{table}[h!]
\caption{In this table, we give the decomposition of $q_\text{PR2} \in \qset{1}^{\times 2}$ as a convex combination of the vertices of $\mdlset{2}$ at $h^2 = 1/10$. It is represented in terms of local deterministic strategies and their corresponding weights $w_\lambda$. Each of these local deterministic strategies is a 6-mismatch strategy, and thus for $h^2 = 1/10$, there is only one way to assign values to $P(\vec{x} \vec{y} | \lambda)$ such that it produces a vertex of $\mdlset{2}$ compatible with $q_\text{PR2}$. The local deterministic strategies are described by specifying Alice's (respectively, Bob's) outputs for the 4 possible block-2 inputs.}
\begin{center}
\def\arraystretch{1.5}
\setlength\tabcolsep{.3cm}
\begin{tabular}{|c|c|c|}
\hline
{\bfseries \boldmath Weight $w_\lambda$} & {\parbox{4.5cm}{$\left. \right.$ \\[.02cm] \bfseries \boldmath Alice's outputs $\vec{a}$ for \\ inputs $\vec{x} = 00, 01, 10, 11$ } } & {\parbox{4.5cm}{$\left. \right.$ \\[.02cm] \bfseries \boldmath Bob's outputs $\vec{b}$ for \\ inputs $\vec{y} = 00, 01, 10, 11$ } } \\[.4cm]
\hline
$1/96$ & $00, 01, 10, 01$ & $01, 00, 00, 00$ \\
\hline
$1/96$ & $00, 10, 10, 10$ & $10, 11, 00, 00$ \\
\hline
$1/96$ & $01, 00, 11, 00$ & $00, 01, 01, 01$ \\
\hline
$1/96$ & $01, 11, 11, 11$ & $11, 10, 01, 01$ \\
\hline
$1/96$ & $10, 00, 00, 00$ & $00, 01, 10, 10$ \\
\hline
$1/96$ & $10, 11, 00, 11$ & $11, 10, 10, 10$ \\
\hline
$1/96$ & $11, 01, 01, 01$ & $01, 00, 11, 11$ \\
\hline
$1/96$ & $11, 10, 01, 10$ & $10, 11, 11, 11$ \\
\hline
$1/48$ & $00, 01, 00, 00$ & $00, 00, 10, 00$ \\
\hline
$1/48$ & $01, 00, 01, 01$ & $01, 01, 11, 01$ \\
\hline
$1/48$ & $10, 11, 01, 01$ & $01, 10, 11, 10$ \\
\hline
$1/48$ & $11, 10, 00, 00$ & $00, 11, 10, 11$ \\
\hline
$1/48$ & $00, 01, 11, 11$ & $11, 00, 01, 00$ \\
\hline
$1/48$ & $01, 00, 10, 10$ & $10, 01, 00, 01$ \\
\hline
$1/48$ & $10, 11, 10, 10$ & $10, 10, 00, 10$ \\
\hline
$1/48$ & $11, 10, 11, 11$ & $11, 11, 01, 11$ \\
\hline
$1/24$ & $00, 00, 01, 00$ & $00, 01, 00, 11$ \\
\hline
$1/24$ & $01, 01, 00, 01$ & $01, 00, 01, 10$ \\
\hline
$1/24$ & $10, 10, 00, 01$ & $10, 00, 10, 10$ \\
\hline
$1/24$ & $11, 11, 01, 00$ & $11, 01, 11, 11$ \\
\hline
$1/24$ & $00, 00, 10, 11$ & $00, 10, 00, 00$ \\
\hline
$1/24$ & $01, 01, 11, 10$ & $01, 11, 01, 01$ \\
\hline
$1/24$ & $10, 10, 11, 10$ & $10, 11, 10, 01$ \\
\hline
$1/24$ & $11, 11, 10, 11$ & $11, 10, 11, 00$ \\
\hline
$5/96$ & $00, 01, 10, 10$ & $10, 00, 00, 00$ \\
\hline
$5/96$ & $01, 00, 00, 00$ & $00, 01, 10, 01$ \\
\hline
$5/96$ & $10, 10, 10, 00$ & $10, 10, 10, 11$ \\
\hline
$5/96$ & $11, 00, 11, 10$ & $11, 11, 00, 01$ \\
\hline
$5/96$ & $00, 11, 00, 01$ & $00, 00, 11, 10$ \\
\hline
$5/96$ & $01, 01, 01, 11$ & $01, 01, 01, 00$ \\
\hline
$5/96$ & $10, 11, 11, 11$ & $11, 10, 01, 10$ \\
\hline
$5/96$ & $11, 10, 01, 01$ & $01, 11, 11, 11$ \\
\hline
\end{tabular}
\end{center}
\def\arraystretch{1}
\end{table}

\clearpage 

\end{document}